\documentclass[letterpaper, 10 pt, conference]{ieeeconf}     % Use this line for a4 paper

\IEEEoverridecommandlockouts                              % This command is only needed if 
\usepackage{cite}
\usepackage{epsfig} % for postscript graphics files
\usepackage{times} % assumes new font selection scheme installed
\usepackage{amsmath,amssymb,amsfonts,bm,mathrsfs}
\usepackage{graphicx}
\usepackage[capitalise]{cleveref}% cref
\usepackage{float}
\usepackage{fp,tikz,pgfplots}
\usepackage[english]{babel}
\usepackage{algorithm}
\usepackage{algpseudocode}
\usepackage{verbatim}
\usepackage{makecell}
\usepackage{multirow}
\usepackage{lipsum} 

\usepackage[shortcuts,acronym]{glossaries} %\gls{...}
\usepackage{glossaries} %\gls{...}
\usepackage{xcolor}
\usetikzlibrary{arrows,shapes,backgrounds,patterns,fadings,matrix,arrows,calc,
	intersections,decorations.markings,
	positioning,arrows.meta}
\usepgfplotslibrary{fillbetween}
\usepgfplotslibrary{statistics}
\pgfplotsset{width=5\columnwidth /5, compat = 1.13,
	height = 60\columnwidth /100, grid= major,
	legend cell align = left, ticklabel style = {font=\scriptsize},
	every axis label/.append style={font=\small},
	legend style = {font={\scriptsize}},title style={yshift=-7pt, font = \small} }

\newacronym{mas}{MAS}{multi-agent system}
\newacronym{ml}{ML}{Machine Learning}

\graphicspath{{fig/}}

\newtheorem{assumption}{\bf Assumption}
\newtheorem{theorem}{\bf Theorem}

\newtheorem{lemma}{\bf Lemma}
\newtheorem{remark}{\bf Remark}
\newtheorem{definition}{\bf Definition}

\newtheorem{property}{\bf Property}

\xdefinecolor{blueZ}{RGB}{0,77,163}
\xdefinecolor{redZ}{RGB}{232,61,117}

\newcommand{\ProofEndSymbol}{\hfill$\blacksquare$}

\definecolor{tab10_blue}{RGB}{31,119,180}
\definecolor{tab10_orange}{RGB}{255,127,14}
\definecolor{tab10_green}{RGB}{44,160,44}
\definecolor{tab10_red}{RGB}{214,39,40}

\definecolor{set_red}{RGB}{228,26,28}
\definecolor{set_blue}{RGB}{55,126,184}
\definecolor{set_green}{RGB}{77,175,74}

\definecolor{julia_blue}{RGB}{0,154,250}
\definecolor{julia_orange}{RGB}{227,111,71}
\definecolor{julia_green}{RGB}{62,164,78}
\definecolor{julia_red}{RGB}{195,113,210}

\newif\ifZeno
\Zenotrue
\title{\LARGE \bf
	Decentralized Event-Triggered Online Learning for Safe Consensus \\
	of Multi-Agent Systems with Gaussian Process Regression
}

\author{Xiaobing Dai$^{*1}$, Zewen Yang$^{*2}$, Mengtian Xu$^1$,  Fangzhou Liu$^3$, Georges Hattab$^{2,4}$ and Sandra Hirche$^1$% <-this % stops a space
    \thanks{*Equal contribution. }
	\thanks{
    This work has been financially supported by the Federal Ministry of Education and Research of Germany in the programme of ``Souverän. Digital. Vernetzt.'' under joint project 6G-life with project identification number: 16KISK002 and by the Germany Federal Ministry of Health (BMG) under grant No. 2523DAT400 (project ``AI-assisted analysis and visualization of pandemic situations'' | AI-DAVis-PANDEMICS).
    }% <-this % stops a space
	\thanks{
		$^{1}$Chair of Information-oriented Control, TUM School of Computation, Information and Technology, Technical University of Munich, 80333 Munich, Germany.
		{\tt\small xiaobing.dai; mengtian.xu; hirche@tum.de}
		$^{2}$Center for Artificial Intelligence in Public Health Research (ZKI-PH), Robert Koch Institute, Berlin, Germany.
		{\tt\small yangz; hattabg@rki.de}
		$^{3}$National Key Laboratory of Modeling and Simulation for Complex Systems, Harbin Institute of Technology, 150001, China. 
        {\tt\small fangzhou.liu@hit.edu.cn}
		$^{4}$Department of Mathematics and Computer Science, Freie Universität Berlin, Berlin, Germany.}
}

\begin{document}

	\maketitle
	\thispagestyle{empty}
	\pagestyle{empty}

	%%%%%%%%%%%%%%%%%%%%%%%%%%%%%%%%%%%%%%%%%%%%%%%%%%%%%%%%%%%%%%%%%%%%%%%%%%%%%%%%
	\begin{abstract}
		Consensus control in multi-agent systems has received significant attention and practical implementation across various domains. However, managing consensus control under unknown dynamics remains a significant challenge for control design due to system uncertainties and environmental disturbances. This paper presents a novel learning-based distributed control law, augmented by an auxiliary dynamics. Gaussian processes are harnessed to compensate for the unknown components of the multi-agent system. For continuous enhancement in predictive performance of Gaussian process model, a data-efficient online learning strategy with a decentralized event-triggered mechanism is proposed. Furthermore, the control performance of the proposed approach is ensured via the Lyapunov theory, based on a probabilistic guarantee for prediction error bounds. To demonstrate the efficacy of the proposed learning-based controller, a comparative analysis is conducted, contrasting it with both conventional distributed control laws and offline learning methodologies.
	\end{abstract}

	%%%%%%%%%%%%%%%%%%%%%%%%%%%%%%%%%%%%%%%%%%%%%%%%%%%%%%%%%%%%%%%%%%%%%%%%%%%%%%%%
	\section{Introduction}
	In recent decades, multi-agent systems (MASs) have garnered considerable interest due to their versatile applications, encompassing domains such as aerial drones, ground vehicles, and underwater crafts \cite{chaikalis2023decentralized, shi2018distributed, 8867291}. Given the inherent challenges associated with uncertainties arising from environmental disturbances and inaccurate system dynamics, learning-based control approaches have emerged as a promising avenue for mitigating these issues. These approaches harness machine learning methodologies to estimate unknown components in the system dynamics, drawing from observed data. Subsequently, these data-driven models are seamlessly integrated into the control systems to enhance their adaptability.
	
	The substantial progress witnessed in the field of deep learning, particularly its capacity to model intricate mappings from data, has unequivocally propelled the adoption of neural networks (NNs) for regression tasks. 
	It is evident that the widespread success of NNs in diverse applications has instigated their integration into the domain of consensus control in MASs \cite{olfati-saberConsensusProblemsNetworks2004, zhangConsensusControlMultiple2019}. 
	For instance, an adaptive protocol that incorporates data-driven models acquired through NNs has been proposed in \cite{guanghouDecentralizedRobustAdaptive2009, caiAdaptiveFiniteTime2016}.
	Furthermore, research efforts have delved into security-related issues, including scenarios involving attacks and actuator faults within compromised networks \cite{lvFullyDistributedAdaptive2020, jinAdaptiveNNBasedConsensus2021}. 
	However, as the number of hidden layers and neurons in deep learning models increases, the adjustment of model weights becomes a time-consuming process. This poses a significant practical challenge for real-time implementation in control systems, primarily because the time delay associated with the learning process adversely affects control performance \cite{dai2023can}.
	Additionally, the integration of NNs does not inherently provide precise prediction guarantees, a limitation that hinders their application in safety-critical domains.
	
	In contrast, Gaussian process (GP) regression has emerged as an alternative approach that carries certain advantages in the context of MAS control. 
	GP provides a probabilistic framework that facilitates principled uncertainty quantification, offering a valuable tool for modeling uncertainties within the MAS and enabling safer control strategies \cite{rasmussenGaussianProcessesMachine2006}. 
	GP regression is first introduced for cooperative learning-based control in MASs in \cite{yangDistributedLearningConsensus2021,yangCooperativeLearning2024,yangAAMAS,yangAAMASextend}, which conceives the novel idea of treating each agent as an individual GP model. 
	However, the aggregation strategy requires the computation of the predictions for neighbor agents, resulting in a larger demand for local computational resources.
	To reduce the computational burden, subsequently, a dynamic average consensus algorithm \cite{kiaTutorialDynamicAverage2019} is integrated to enhance the posterior mean and variance of GPs with only prediction on the local states for each agent \cite{ledererCooperativeControlUncertain2023}.
	But the aforementioned works overlook the imperative consideration of online learning with MAS control, leading to a deficiency in the ability to update GP models and a restricted efficacy in prediction.
	To address this issue, an online learning approach by adding new training samples was proposed for MAS flocking control in \cite{beckersLearningRigiditybasedFlocking2022}. 
	However, it only considers the data collection in a fixed frequency with the deletion strategy for data storage management and lacks the consideration of the different importance of the data.
	To tackle this problem, an event-triggered learning approach for formation control of MASs is presented in \cite{dai2024cooperative} to choose only valuable data for GP-based control. 
	Nevertheless, each agent requires access to global information such as communication topologies, which may not be feasible in certain circumstances.
		
	In this paper, we consider the average consensus for a multi-agent system with unknown dynamics.
	A novel distributed learning-based consensus controller featuring an auxiliary dynamics is proposed, which isolates the influence of the dynamics uncertainties from the estimation of the average states.
	Moreover, Gaussian process regression is employed to compensate for the unknown dynamics.
	For better control performance by using the GP-based controller, a decentralized event-triggered online learning strategy is introduced for control performance improvement with efficient data storage management. 
	Furthermore, a stability analysis of the MAS is included with the probabilistic prediction guarantees for average consensus error, further enhancing its operational effectiveness for safe learning-based control tasks.
	To underscore the superiority of the proposed approach, numerical simulations are carried out demonstrating its effectiveness through a comparative analysis against conventional control and learning strategies.

%	The remaining content of this paper is structured as follows. 
%	\cref{section_problem_setting} presents the problem setting and proposes the distributed consensus control law. 
%	In \cref{section_distributed_event_trigger}, the development of the decentralized event-triggered strategy for online learning is provided with control performance guarantee. 
%	\cref{section_simulation} encompasses the execution of numerical simulations, followed by the conclusion of this paper in \cref{section_conclusion}. \looseness=-1
	%%%%%%%%%%%%%%%%%%%%%%%%%%%%%%%%%%%%%%%%%%%%%%%%%%%%%%%%%%%%%%%%%%%%%%%%%%%%%%%%
	\section{Problem Setting and Preliminaries} \label{section_problem_setting}
	
	In this paper, a MAS with $N \in \mathbb{N}_+$ homogeneous agents is considered, whose dynamics is expressed by 
	\begin{align} \label{eqn_dynamics}
		\dot{x}_i = h(x_i) + f(x_i) + g(x_i) u_i, ~ \forall i = 1,\cdots, N,
	\end{align}
	where $x_i \in \mathbb{X} \subset \mathbb{R}$ with a compact domain $\mathbb{X}$ and $u_i \in \mathbb{R}$ denote the system state and control input, respectively.  
	For notational simplicity, we focus on scalar agent dynamics, whose results can be directly extended to multi-dimensional cases by using multi-input-output machine learning tools and the Kronecker operator. 
	The functions $h: \mathbb{X} \to \mathbb{R}$ and $g: \mathbb{X} \to \mathbb{R}$ describe the known parts of the individual dynamics of the agents, while the function $f: \mathbb{X} \to \mathbb{R}$ is regarded as unknown components such as unmodeled dynamics and environmental uncertainties.
	To ensure the controllability for each agent, the following assumption on $g(\cdot)$ is required.
	\begin{assumption} \label{assumption_g}
		For any $x_i \in \mathbb{X}$, it has $g(x_i) \ne 0$.
	\end{assumption}
	The non-singularity of $g(\cdot)$ is a common prerequisite for the controller non-autonomous system, such that the control input $u_i$ is effective at any state.
	While this assumption induces restrictions on the system class, the focus of this paper is on the design of event-triggered learning-based controllers.
	The extension to other system classes is left in the future work.
	
	The agents communicate through a network, whose topology is presented by a time-invariant undirected connected graph $\mathcal{G} = \{ \mathcal{V}, \mathcal{E}\}$, where $\mathcal{V} = \{ 1, \cdots, N \}$ is the vertex set representing the indices of agents and $\mathcal{E} \in \mathcal{V} \times \mathcal{V}$ denotes the edge set. 
	The topology of $\mathcal{G}$ is characterized by an adjacency matrix $\bm{\mathcal{A}} = [a_{ij}]_{i,j \in \mathcal{V}} \in \mathbb{R}^{N \times N}$. The entry $a_{ij} = a_{ji} = 1$, when there exists a communication channel such that the agent $i$ can send and receive the information from agent $j$, i.e., $(i,j) \in \mathcal{E}$, otherwise $a_{ij} =a_{ji} = 0, \forall i,j \in \mathcal{V}$. The set $\mathcal{N}_i$ of agent $i$ contains all the neighboring agents $j \in \mathcal{V}$ indicating $a_{ij} = 1$.
	The Laplacian matrix of $\mathcal{G}$ is defined as $\bm{\mathcal{L}} = [l_{ij} ]_{i,j \in \mathcal{V}} \in \mathbb{R}^{N \times N}$ with $l_{ii} = \sum_{j=1}^N a_{ij}$ and $l_{ij} = - a_{ij}$ for $j \ne i$.
	
	The control task is to achieve average consensus, i.e., the agent's state $x_i$ tends to reach $\bar{x}^* = N^{-1} \bm{1}_N^T \bm{x}(0)$ for $i = 1, \cdots, N$, where $\bm{x} = [x_1, \cdots, x_N]^T$ denotes the concatenated state. 
	However, due to the unknown function $f(\cdot)$, the asymptotic stability of the MAS is not expected. Therefore, in this paper, we investigate $\epsilon$-average consensus instead, whose definition is given as follows.
	\begin{definition} \label{definition_consensus}
		If there exists a well-defined constant $\epsilon \in \mathbb{R}_{0,+}$, such that $\lim_{t \to \infty} \| \bm{x}(t) - \bm{1}_N \bar{x}^* \| \le \epsilon$.
		Then, the multi-agent system achieves $\epsilon$-average consensus.
	\end{definition}

	To achieve such a control task, we firstly consider the conventional controller using consensus law $r_i$ from \cite{saber2003consensus} as \looseness=-1
    \begin{subequations} \label{eqn_conventional_control_law}
        \begin{align}
            &u_i = - \frac{1}{g(x_i)} \left( h(x_i) + \hat{f}_i(x_i) + c r_i \right), \\
            &r_i = \sum_{j \in \mathcal{N}_i} (x_i - x_j),
        \end{align}	
    \end{subequations}
	where $c \!\in\! \mathbb{R}_+$ and $\hat{f}_i(x_i)$ is the prediction of unknown $f(x_i)$ from agent $i$.
	While the controlled system achieves consensus by using \eqref{eqn_conventional_control_law} under dynamics uncertainties from unknown $f_i(\cdot)$ proven in \cite{li2014distributed}, the state of each agent is not guaranteed in the compact domain $\mathbb{X}$ and not close to $\bar{x}^*$.
	Without rigorous proof, an example is shown in Appendix reflecting the potential instability of the controlled system with \eqref{eqn_conventional_control_law}.
	To address this problem, a novel distributed control law is proposed integrating an auxiliary dynamics as follows \looseness=-1
\begin{subequations}\label{eqn_controller}
		\begin{align}
			&u_i = - \frac{1}{g(x_i)} \left( h(x_i) + \hat{f}_i(x_i) + c r_i - \dot{\bar{x}}_i \right), \\
			&r_i = \sum_{j \in \mathcal{N}_i} (\tilde{x}_i - \tilde{x}_j) + \tilde{x}_i,
		\end{align}
	\end{subequations}
	where $c$ is the control gain and the error between the current state $x_i$ and auxiliary state $\bar{x}_i$ is denoted by $\tilde{x}_i = x_i - \bar{x}_i$.
	The auxiliary state $\bar{x}_i$ of the agent $i$ is obtained by solving the following dynamics as
	\begin{align} \label{eqn_observer_dynamics}
		\dot{\bar{x}}_i = - \bar{c} \sum_{j \in \mathcal{N}_i} (\bar{x}_i - \bar{x}_j), ~\forall i = 1, \cdots, N,
	\end{align}
	with gain $\bar{c} \in \mathbb{R}_+$ and the initial condition $\bar{x}_i(0) = x_i(0)$.
	
	The estimation function $\hat{f}_i(\cdot)$ of the unknown function $f(\cdot)$ in \eqref{eqn_controller} associated with the agent $i$ is obtained from data-driven method using the individual data set $\mathbb{D}_i = \{ x_i^{(\iota)}, y_i^{(\iota)} \}_{\iota \in \mathbb{N}}$ of each agent $i \in \mathcal{V}$, which satisfies the following assumption.
	\begin{assumption} \label{assumption_dataset}
		The data pair $\{ x_i^{(\iota)}, y_i^{(\iota)} \}$ is available on each agent $i \in \mathcal{V}$ at any time instance $t_i^{(\iota)} \in \mathbb{R}_{0,+}$ with $(\iota) \in \mathbb{N}$, where the measured output denotes $y_i^{(\iota)} = \dot{x}_i(t_i^{(\iota)}) - h(x_i(t_i^{(\iota)})) \!-\! g(x_i(t_i^{(\iota)})) u_i(t_i^{(\iota)}) + w_i^{(\iota)}$.
		The measurement noise $w_i^{(\iota)} \in \mathbb{R}$ follows a zero-mean, independent and identical Gaussian distribution with variance $\sigma_{n}^2$ and $\sigma_{n} \in \mathbb{R}_+$.
	\end{assumption}
	\cref{assumption_dataset} indicates each agent can collect the data individually.
	Additionally, this assumption accommodates the Gaussian noise in $y_i^{(\iota)}$, which often stems from the numerical approximation of $\dot{x}_i(t_i^{(\iota)})$, such as the utilization of finite difference methods. 
	The relaxation of the noise distribution exists as discussed in works such as \cite{chowdhury2017kernelized, maddalena2021deterministic}. 
	Moreover, \cref{assumption_dataset} provides the possibility for the continuous enhancement of $\hat{f}_i(\cdot)$ through ongoing online data collection.
	
	In order to collect the necessary data pairs individually on each agent with high data efficiency, we introduce a decentralized event-triggered online learning strategy. 
	This strategy ensures that the pair ${ x_i^{(\iota)}, y_i^{(\iota)} }$ will be incorporated into the training dataset $\mathbb{D}_i$ of agent $i$ at time $t_i^{(\iota)}$ in accordance with the following condition
	\begin{align} \label{eqn_general_trigger_formulation}
		t_i^{(\iota)} = \min \left\{ t_i : t_i > t_i^{(\iota-1)} ~\wedge~ \rho(x_i(t_i), \bar{x}_i(t_i)) > 0 \right\},
	\end{align}
	where $\rho(\cdot): \mathbb{X} \times \mathbb{X} \to \mathbb{R}$ indicates the universal trigger function for all agents by evaluated by using individual state information including the auxiliary state $\bar{x}_i$.
	The detailed expression of $\rho(\cdot)$ will be presented in \cref{section_distributed_event_trigger}.
	
	%%%%%%%%%%%%%%%%%%%%%%%%%%%%%%%%%%%%%%%%%%%%%%%%%%%%%%%%%%%%%%%%%%%%%%%%%%%%%%%%\
	\section{Decentralized Event-Triggered Online Learning with Gaussian Processes} \label{section_distributed_event_trigger}
	
	\subsection{Gaussian Process Regression}
	
	To infer the unknown $f(\cdot)$, Gaussian process regression is adopted, which induces a distribution of $f(\cdot)$ by the prior mean $m(\cdot): \mathbb{X} \to \mathbb{R}$ and kernel function $\kappa(\cdot, \cdot): \mathbb{X} \times \mathbb{X} \to \mathbb{R}_{0,+}$, i.e., $f(\cdot) \sim \mathcal{GP}(m(\cdot), \kappa(\cdot,\cdot))$.
	The prior mean $m(\cdot)$ reflects the known part of $f(\cdot)$, which for the dynamics in \eqref{eqn_dynamics} is encoded in $h(\cdot)$.
	Therefore, $m(\cdot) = 0$ is set.
	The kernel function $\kappa(\cdot)$ indicates the covariance between two sample points and satisfies the following assumption.
	\begin{assumption} \label{assumption_GP}
		The continuous function $f(\cdot)$ is a sample from a Gaussian process $\mathcal{GP}(0, \kappa(\cdot,\cdot))$ with a Lipschitz kernel $\kappa(|x - x'|) = \kappa(x,x')$ w.r.t $|x - x'|$.
	\end{assumption}
	From this assumption, the prior distribution of the unknown function $f(\cdot)$ is defined.
    	Moreover, the assumption on the kernel $\kappa(\cdot, \cdot)$ can be achieved by choosing any continuous function reflecting the continuity of $f(\cdot)$, and then its Lipschitz continuity is derived by considering the compact input domain $\mathbb{X}$.
	And the commonly used kernels, such as square exponential kernel, satisfy the Lipschitz condition.
	Therefore, this assumption is practically not restrictive.
	
	Consider an individual data set $\mathbb{D}_i$ on agent $i \in \mathcal{V}$ containing $M \in \mathbb{N}$ training samples and satisfying \cref{assumption_dataset}, the posterior mean $\mu_i(\cdot)$ and variance $\sigma_i(\cdot)$ are obtained by using the Gaussian process model with \cref{assumption_GP} as
	\begin{align}
		&\mu_i(x_i) = \bm{k}_{X,i}^T (\bm{K}_i + \sigma_n^2 \bm{I}_M)^{-1} \bm{y}, \\
		&\sigma_i^2(x_i) = \kappa(0) - \bm{k}_{X,i}^T (\bm{K}_i + \sigma_n^2 \bm{I}_M)^{-1} \bm{k}_{X,i},
	\end{align}
	where $\bm{k}_{X,i} = [\kappa(x_i, x_i^{(1)}), \cdots, \kappa(x_i, x_i^{(M)})]^T$, $\bm{K}_i = [\kappa(x_i^{(p)}, x_i^{(q)})]_{p,q = 1, \cdots, M}$ and $\bm{y} = [y_i^{(i)}, \cdots, y_i^{(M)}]^T$.
	While the obtained posterior mean $\mu_i(\cdot)$ is adopted as the compensation of the unknown function $f(\cdot)$ in \eqref{eqn_controller}, i.e., $\hat{f}_i(\cdot) = \mu_i(\cdot)$, the posterior variance $\sigma_i(\cdot)$ is used to quantify the prediction error shown as follows.
	\begin{lemma} [\! \cite{lederer2019uniform}] \label{lemma_GP_bound}
		Consider the Gaussian process model on agent $i$ for the inference of an unknown function $f(\cdot)$ satisfying \cref{assumption_GP}, which is based on a data set under \cref{assumption_dataset}.
		Choose $\delta \in (0,1) \subset{R}$ and $\tau_i \in \mathbb{R}_+$ such that $\gamma_i \le \min_{x_i \in \mathbb{X}} \sqrt{\beta} \sigma_i(x_i)$ with
		\begin{align}
			&\beta = 2 \log \left( \frac{1}{2 \delta \tau_i} \Big(\max_{x \in \mathbb{X}} x - \min_{x \in \mathbb{X}} x \Big) + \frac{1}{\delta} \right), \\
			&\gamma_i = (L_f + L_{\mu,i} + \sqrt{\beta} L_{\sigma}) \tau_i,
		\end{align}
		where $L_f$, $L_{\mu,i}$, $L_{\sigma}$ are the Lipschitz constants for $f(\cdot)$, $\mu_i(\cdot)$, $\sigma_i(\cdot)$, respectively.
		Then, the prediction error is uniformly bounded by
        \begin{align}
			|f(x_i) - \mu_i(x_i)| \le \eta_i(x_i) = 2 \sqrt{\beta} \sigma_i(x_i), ~ \forall x_i \in \mathbb{X},
		\end{align}
        with probability of at least $1 - \delta$.
	\end{lemma}
	Although with conservatism, \cref{lemma_GP_bound} provides a numerical way to probabilistically bound the prediction error, which is commonly used for the control in safety-critical scenarios.
	Note that \cref{lemma_GP_bound} requires the Lipschitz continuity of the unknown function $f(\cdot)$ represented by $L_f$, which is easily shown by considering the continuity of $f(\cdot)$ in \cref{assumption_GP} and the compact domain $\mathbb{X}$.
	The approximation of $L_f$ can be obtained by using the data-driven methods as in \cite{lederer2019uniform}.
	Moreover, the Lipschitz constants $L_{\mu,i}$ and $L_{\sigma}$ exist due to the Lipschitz kernel function $\kappa(\cdot)$ in \cref{assumption_GP}, whose explicit computation follows \cite{yin2023learningbased}.
	
	\subsection{Decentralized Event-Triggered Learning} \label{subsection_event_trigger}
	
	Since the performance of a learning based controller incorporated with the data-driven model $\hat{f}(\cdot)$ in \eqref{eqn_controller} is contingent upon the precision of its predictions \cite{lederer2019uniform}, it is common to employ online learning techniques to enhance prediction accuracy for better control performance. 
	To facilitate effective and efficient online learning, a decentralized event-trigger mechanism is designed to improve the efficiency of the learning process. 
	This mechanism induces a careful selection of the training data relying on a trigger function denoted as $\rho(\cdot)$ in \eqref{eqn_general_trigger_formulation}.
	The specific choice of this trigger function is a vital aspect of our approach, aimed at achieving data efficiency and improving the prediction performance of the system during real-time operation.
	\begin{align} \label{eqn_trigger_function}
		\rho(x_i, \bar{x}_i) \!=\! \eta_i(x_i) \!-\! \max \left\{ c |x_i \!-\! \bar{x}_i| \!-\! \sqrt{N-1} \underline{\eta} , \underline{\eta} \right\},
	\end{align}
	where $\underline{\eta} = 2 \sqrt{\beta} \sigma_n$.
	The trigger function $\rho(\cdot)$ is composed of three parts, namely the current local prediction error $\eta_i(x_i)$ defined in \cref{lemma_GP_bound}, the constant term $\sqrt{N-1} \underline{\eta}$ and the error between actual state $x_i$ and auxiliary state $\bar{x}_i$, i.e., $\tilde{x}_i$.
    Note that the time stamp $t_i$ in the trigger function $\rho(\cdot)$ in \eqref{eqn_general_trigger_formulation} is dropped in \eqref{eqn_trigger_function} for notational simplicity.
	
	\begin{remark}
		While based on the distributed control as in \eqref{eqn_controller}, the operations for GP models on each agent, including prediction and online learning, are only based on the local information, i.e., $x_i$ and $\bar{x}_i$.
		Therefore, the proposed event-trigger is decentralized similarly as in \cite{yang2016decentralized}.
		Moreover, the distributed control law \eqref{eqn_controller} with auxiliary dynamics \eqref{eqn_observer_dynamics} and the decentralized event-trigger \eqref{eqn_general_trigger_formulation} only require the exchange of system state $x_i$ and $\bar{x}_i$ among agents, which reduces the communication burden in network compared with \cite{lederer2022cooperative}.
		Furthermore, it requires no computation for the predictions of the neighbors, decreasing the local computational demand on each agent compared to \cite{yangDistributedLearningConsensus2021}.
	\end{remark}
	
	\begin{remark}
		The constant $\sqrt{N-1} \underline{\eta}$ is introduced, due to the lack of information about actions of online learning on other agents.
		Moreover, by considering its non-negativity because of $N \ge 1$ and $\sigma_n > 0$ as in \cref{assumption_dataset}, the event-trigger becomes more conservative than single agent case in \cite{dai2023can}, i.e., inducing more trigger events.
		Moreover, the conservatism grows with the number of agents $N$, which is intuitive due to more unknown actions on more agents.
		In particular for a single agent system as in \cite{umlauft2019feedback}, the trigger function $\rho(\cdot)$ in \eqref{eqn_trigger_function} is degenerated to $\rho(x_i, \bar{x}_i) = \eta_i(x_i) - \max \{ c |\tilde{x}_i|, \underline{\eta} \}$.
		The degenerated form indicates the conditions $\rho(x_i, \bar{x}_i) > 0$ and $\{\eta_i(x_i) > c |\tilde{x}_i| \vee \eta_i(x_i) > \underline{\eta}\}$ are equivalent and identical as in \cite{umlauft2019feedback, lederer2021gaussian, dai2023can}, whose effects on the learning-based control are proven with ultimately bounded error.
	\end{remark}
	
	When the trigger condition \eqref{eqn_general_trigger_formulation} is satisfied on agent $i$, the GP model and the prediction $\mu_i(x_i)$ on agent $i$ will be updated by adding new training data.
	We denote the updated posterior mean, variance, and the corresponding prediction error by $\mu_i^+(x_i)$, $\sigma_i^+(x_i)$, and $\eta_i^+(x_i)$, respectively, using the updated data set satisfying
	\begin{align} \label{eqn_prediction_error_bound_updated}
		|f(x_i) - \mu_i^+(x_i) | \le \eta_i^+(x_i) \le \underline{\eta}, ~\forall x_i \in \mathbb{X},
	\end{align}
    with a probability of at least $1 \!-\! \delta$, which is proven by considering the $1$-point upper bound $\sigma_{i,\text{1pt}}^+$ of variance \cite{williams2000upper} as \looseness=-1
	\begin{align} \label{eqn_sigma_after_update}
		( \sigma_{i,\text{1pt}}^+ )^2 = \kappa(0) - \frac{\kappa(0)^2}{\kappa(0) + \sigma_n^2} = \frac{\kappa(0) \sigma_n^2}{\kappa(0) + \sigma_n^2} \le \sigma_n^2
	\end{align}
	and the fact that more data results in smaller posterior variance \cite{lederer2021gaussian}, i.e., $\sigma_i^+(x_i) \le \sigma_{i,\text{1pt}}^+ \le \sigma_n$.
    Combining with the definition of the prediction error bound in \cref{lemma_GP_bound}, the upper bound of $\eta_i^+(x_i)$ is derived as in \eqref{eqn_prediction_error_bound_updated}.
	For notational simplicity, define $\hat{\mu}_i(x_i) = \mu_i(x_i)$ and $\hat{\eta}_i(x_i) = \eta_i(x_i)$ for all agent $i$ with $\rho(x_i, \bar{x}_i) \le 0$, as well as $\hat{\mu}_i(x_i) = \mu_i^+(x_i)$ and $\hat{\eta}_i(x_i) = \eta_i^+(x_i)$ for all agent $i$ with $\rho(x_i, \bar{x}_i) > 0$.
	Note that only the prediction error bound without model update, i.e., $\mu_i(x_i)$, is used for the evaluation of $\rho(\cdot)$ in \eqref{eqn_trigger_function}.
	Then, $\hat{f}_i(x_i) = \hat{\mu}_i(x_i)$ and the following properties can be derived by considering the event-trigger \eqref{eqn_general_trigger_formulation} with \eqref{eqn_trigger_function}.
	\begin{property} \label{property_trigger}
		For any agent $i \in \mathcal{V}$ using the event-trigger \eqref{eqn_general_trigger_formulation} with \eqref{eqn_trigger_function}, the following statements hold.
		\begin{enumerate}
			\item If $c |\tilde{x}_i| \le (\sqrt{N-1} + 1) \underline{\eta}$, then $\hat{\eta}_i(x_i) \le \underline{\eta}$.
			\item If $c |\tilde{x}_i| > (\sqrt{N-1} + 1) \underline{\eta}$ and $\rho(x_i, \bar{x}_i) > 0$, then $\hat{\eta}_i(x_i) \le \underline{\eta}$.
			\item If $c |\tilde{x}_i| > (\sqrt{N-1} + 1) \underline{\eta}$ and $\rho(x_i, \bar{x}_i) \le 0$, then
			\begin{align}
				c^2 |\tilde{x}_i|^2 \ge \eta_i^2(x_i) + (N - 1) \underline{\eta}^2.
			\end{align}
		\end{enumerate}
	\end{property}
    % \vspace{0.1cm}
	\begin{proof}
		To prove the first statement, two cases with different sign of $\rho(\cdot)$ are considered.
		The case with $\rho(x_i, \bar{x}_i) > 0$ activates the data collection on agent $i$ and then results in $\hat{\eta}_i(x_i) = \eta_i^+(x_i) \le \underline{\eta}$ by comparing the posterior variance as in \eqref{eqn_sigma_after_update}.
		For the case with $\rho(x_i, \bar{x}_i) \le 0$, no model update occurs on agent $i$, and from the design of $\rho(\cdot)$ in \eqref{eqn_trigger_function} it holds
		\begin{align} \label{eqn_eta_no_trigger_in_bound}
			\hat{\eta}_i(x_i) = \eta_i(x_i) \le c | \tilde{x}_i | - \sqrt{N - 1} \underline{\eta} \le \underline{\eta}.
		\end{align}
		% Note that $\sqrt{N} - \sqrt{N - 1}$ is monotonically decreasing w.r.t $N \in \mathbb{N}_+$, and the maximum achieves when $N = 1$, i.e., $\max_{N \in \mathbb{N}_+} (\sqrt{N} - \sqrt{N - 1}) = \sqrt{1} - \sqrt{1 - 1} = 1$.
		% Then, \eqref{eqn_eta_no_trigger_in_bound} is reformulated as $\hat{\eta}_i(x_i) \le \underline{\eta}$.
		Combining the results from these two cases, the first statement is proven.
		
		The second statement is directly obtained by considering that the model update is activated at agent $i$, since $\rho(x_i, \bar{x}_i) > 0$.
		Therefore, the boundness of the prediction error $\hat{\eta}_i(x_i) \!=\! \eta_i^+(x_i)$ by $\underline{\eta}$ is derived by using \eqref{eqn_sigma_after_update}.
		
		The proof of the third statement considers the non-positivity of $\rho(x_i, \bar{x}_i)$, where $\max \left\{ c |\tilde{x}_i| - \sqrt{N-1} \underline{\eta} , \underline{\eta} \right\} = c |\tilde{x}_i| - \sqrt{N-1} \underline{\eta}$.
        Then, $\rho(x_i, \bar{x}_i) < 0$ is reformulated as
		\begin{align}
			c^2 |\tilde{x}_i|^2 \!\ge\! ( \eta_i(x_i) \!+\! \sqrt{N \!-\! 1} \underline{\eta})^2 \!\ge\! \eta_i^2(x_i) \!+\! (N \!-\! 1) \underline{\eta}^2,
		\end{align}
		where the second inequality is derived by considering the positivity of $\eta_i(x_i)$ and $\underline{\eta}$ under $N \ge 1$.
	\end{proof}
	
	\begin{remark}
        % The maximum operator in \eqref{eqn_trigger_function} ensures the online learning only occurs when $\eta_i(x_i) > \underline{\eta}$, such that the updated prediction error bound satisfies $\eta_i^+(x_i) \le \underline{\eta}$.
        % Note that the results in \cref{property_trigger} still hold with this maximum operator in \eqref{eqn_trigger_function}, but saving the unnecessary data collection and inducing high data efficiency.
        The maximum operator in \eqref{eqn_trigger_function} ensures the online learning only occurs when $\eta_i(x_i) > \underline{\eta}$, which saves the unnecessary data collection and induces high data efficiency.
	\end{remark}
	
	Note that \cref{property_trigger} only shows necessary statements for further derivation, i.e., the analysis of control performance, which is shown in \cref{subsection_performance}.
	
	\subsection{Consensus Performance Analysis} \label{subsection_performance}
	
	To analyze the control performance under distributed control law \eqref{eqn_controller} with auxiliary dynamics \eqref{eqn_observer_dynamics} and event-triggered online learning \eqref{eqn_trigger_function}, the error between the state $x_i$ of agent $i$ and their ideal average $\bar{x}^*$ defined in \cref{definition_consensus} is studied, which is shown as follows.
	\begin{theorem} \label{theorem_consensus_error}
		Consider a MAS containing $N$ agents under a connected undirected topology $\mathcal{G}$ with individual dynamics \eqref{eqn_dynamics} satisfying \cref{assumption_g}, which is controlled by \eqref{eqn_controller} with \eqref{eqn_observer_dynamics}.
		To infer the unknown function $f(\cdot)$ in \eqref{eqn_controller}, a GP model satisfying \cref{assumption_GP} is employed for each agent with a data set under \cref{assumption_dataset}.
		Moreover, each agent adopts the decentralized event-trigger \eqref{eqn_general_trigger_formulation} and \eqref{eqn_trigger_function} for online learning by adding new data pairs into the individual data set.
		Choose $\delta \in (0, 1 / N)$, then the system achieves probabilistic $\epsilon$-average consensus with $\epsilon = 2 c^{-1} N \underline{\eta}$, i.e., 
		\begin{equation}
			\Pr \left\{ \lim_{t \to \infty}\| \bm{x}(t) - \bm{1}_N \bar{x}^* \| \le \epsilon \right\} \ge 1 - N \delta.
		\end{equation}
	\end{theorem}
	\vspace{0.2cm}
 
	The proof of \cref{theorem_consensus_error} is divided into two parts, namely the consensus of the auxiliary state \eqref{eqn_observer_dynamics} in \cref{lemma_observer} and the ultimate boundness of the consensus error in \cref{lemma_tracking}.
	
	\begin{lemma} \label{lemma_observer}
		Consider a multi-agent system with $N$ agents under a connected undirected graph $\mathcal{G}$, whose auxiliary dynamics is designed as in \eqref{eqn_observer_dynamics}, such that $\dot{\bar{\bm{x}}} = - \bar{c} \bm{\mathcal{L}} \bar{\bm{x}}$ with $\bar{\bm{x}} = [\bar{x}_1, \cdots, \bar{x}_N]^T$.
		Then, all the observed auxiliary states $\bar{x}_i$ for $\forall i \!\in\! \mathcal{V}$ converge to the average state $\bar{x}^*$ recalled from \cref{definition_consensus} as $\bar{x}^* \!=\! N^{\!-\!1} \bm{1}_N^T \bm{x}(0)$, i.e., $\lim_{t \to \infty} \bar{\bm{x}}(t) \!=\! \bm{1}_N \bar{x}^* / N$. \looseness=-1
	\end{lemma}
	
	\textit{Proof sketch:}
	The proof of \cref{lemma_observer} follows \cite{godsil2001algebraic} by showing the negative semi-definite of $\dot{V}_o$ for the positive definite Lyapunov function $V_o = \bar{\bm{x}}^T \bar{\bm{x}}$.
	Combining with the LaSalle's invariance principle, the asymptotic stability is shown with the equilibrium point of $\bar{x}_i$ as $\bar{x}^*$ by considering $\bm{1}_N^T \dot{\bar{\bm{x}}} = 0$ for the connected undirected graph.
	\ProofEndSymbol
	
	\begin{remark}
		Considering the auxiliary state $\bar{x}_i$ is approaching the average state $\bar{x}^*$ due to \eqref{eqn_observer_dynamics}, $\bar{x}_i$ can be regarded as the estimated average state for each agent.
	\end{remark}
	
	Next, we prove the boundness of the consensus error related to $\tilde{\bm{x}} = [\tilde{x}_1, \cdots, \tilde{x}_N]^T$, which is shown as follows.
	
	\begin{lemma} \label{lemma_tracking}
		Let all assumptions in \cref{theorem_consensus_error} hold and choose $\delta \in (0, 1 / N)$, then it holds $\lim_{t \to \infty}\| \tilde{\bm{x}}(t) \| \le 2 c^{-1} N \underline{\eta}$ with probability of at least $1 - N \delta$.
	\end{lemma}
	\begin{proof}
		Apply the control law \eqref{eqn_controller} with \eqref{eqn_observer_dynamics} into \eqref{eqn_dynamics}, the concatenated dynamics of the controlled system denotes
		\begin{align} \label{eqn_dot_tilde_x}
			\dot{\tilde{\bm{x}}} = \dot{\bm{x}} - \dot{\bar{\bm{x}}} = - c (\bm{\mathcal{L}} + \bm{I}_N) \tilde{\bm{x}} + \bm{f}(\bm{x}) - \hat{\bm{\mu}}(\bm{x}),
		\end{align}
		where the concatenated unknown function and its prediction denote $\bm{f}(\bm{x}) = [f(x_1), \cdots, f(x_N)]^T$ and $\hat{\bm{\mu}}(\bm{x}) = [\hat{\mu}_1(x_1), \cdots, \hat{\mu}_N(x_N)]^T$ respectively with $\hat{\mu}_i(\cdot)$ defined in \cref{subsection_event_trigger}.
		Consider the controlled system with event-triggered learning as a switching system, and then the control performance is proven using the common Lyapunov theory.
		Choose the common Lyapunov candidate $V_c$ as
		\begin{align}
			V_c = \tilde{\bm{x}}^T \tilde{\bm{x}},
		\end{align}
		whose time derivative using \eqref{eqn_dot_tilde_x} is written as
		\begin{align}
			\dot{V}_c = - 2 c \tilde{\bm{x}}^T (\bm{\mathcal{L}} + \bm{I}_N) \tilde{\bm{x}} + 2 \tilde{\bm{x}}^T (\bm{f}(\bm{x}) - \hat{\bm{\mu}}(\bm{x})).
		\end{align}
		Consider the fact that $\underline{\lambda}(\bm{\mathcal{L}} + \bm{I}_N) = 1$ \cite{godsil2001algebraic}, where $\underline{\lambda}(\cdot)$ returns the minimal eigenvalue of the matrix.
		Then, the derivative $\dot{V}_c$ is bounded by
		\begin{align} \label{eqn_dotV}
			\dot{V}_c &\le -2 c \| \tilde{\bm{x}} \|^2 + 2 \| \tilde{\bm{x}} \| \| \bm{f}(\bm{x}) - \hat{\bm{\mu}}(\bm{x}) \| \\
			&\le -2 \| \tilde{\bm{x}} \| ( c \| \tilde{\bm{x}} \| - \| \hat{\bm{\eta}}(\bm{x}) \|) \nonumber
		\end{align}
		with a probability of at least $1 - N \delta$ by using the union bound, where $\hat{\bm{\eta}}(\bm{x}) = [\hat{\eta}_1(x_1), \cdots, \hat{\eta}_N(x_n)]^T$.
		Note that the negativity of $\dot{V}_c$ is achieved when
		\begin{align}
			c \| \tilde{\bm{x}} \| > \| \hat{\bm{\eta}}(\bm{x}) \| ~ \Leftrightarrow ~ c^2 \| \tilde{\bm{x}} \|^2 > \| \hat{\bm{\eta}}(\bm{x}) \|^2.
		\end{align}
		Next, according to the value of $\tilde{x}_i$ and $\rho(x_i, \bar{x}_i)$, the agents are divided into $3$ sets defined as
		\begin{align}
			&\mathbb{S}_1 \!=\! \{ i \!\in\! \mathcal{V} \!:\! c |\tilde{x}_i| \!\le\! (\sqrt{N \!-\! 1} \!+\! 1) \underline{\eta} \}, \\
			&\mathbb{S}_2 \!=\! \{ i \!\in\! \mathcal{V} \!:\! c |\tilde{x}_i| \!>\! (\sqrt{N \!-\! 1} \!+\! 1) \underline{\eta} ~\wedge~ \rho(x_i, \bar{x}_i) \!>\! 0 \}, \\
			&\mathbb{S}_3 \!=\! \{ i \!\in\! \mathcal{V} \!:\! c |\tilde{x}_i| \!>\! (\sqrt{N \!-\! 1} \!+\! 1) \underline{\eta} ~\wedge~ \rho(x_i, \bar{x}_i) \!\le\! 0 \}.
		\end{align}
		satisfying $\bigcup_{i = 1,2,3} \mathbb{S}_i = \mathcal{V}$ and $\mathbb{S}_i \cap \mathbb{S}_j = \emptyset$ for $i,j = 1, 2, 3$ and $i \ne j$.
		% Using the results in \cref{property_trigger}, the value of $\| \hat{\bm{\eta}}(\bm{x}) \|^2$ is written as
  %       \begin{align} \label{eqn_eta_bound}
		% 	\| \hat{\bm{\eta}}(\bm{x}) \|^2 &= \sum_{j = 1}^3 \sum_{i \in \mathbb{S}_j} \hat{\eta}_i^2(x_i) \le \sum_{i \in \mathbb{S}_1 \cup \mathbb{S}_2 } \underline{\eta}^2 + \sum_{i \in \mathbb{S}_3} \hat{\eta}_i^2(x_i) \nonumber \\
		% 	&= N \underline{\eta}^2 + \sum_{i \in \mathbb{S}_3} \Big( \hat{\eta}_i^2(x_i) - \underline{\eta}^2 \Big).
		% \end{align}
		% Similarly, the value of $\| \tilde{\bm{x}} \|^2$ is bounded as
		% \begin{align} \label{eqn_c_tilde_x_bound}
		% 	c^2 \| \tilde{\bm{x}} \|^2 &= c^2 \sum_{j = 1}^3 \sum_{i \in \mathbb{S}_j} |\tilde{x}_i|^2 \ge \sum_{i \in \mathbb{S}_2} c^2 |\tilde{x}_i|^2 + \sum_{i \in \mathbb{S}_3} c^2 |\tilde{x}_i|^2 \nonumber \\
		% 	&> \sum_{i \in \mathbb{S}_2} N \underline{\eta}^2 + \sum_{i \in \mathbb{S}_3} \Big( \hat{\eta}_i^2(x_i) + (N - 1) \underline{\eta}^2 \Big) \\
		% 	&= | \mathbb{S}_2 \cup \mathbb{S}_3 | N \underline{\eta}^2 + \sum_{i \in \mathbb{S}_3} \Big( \eta_i^2(x_i) - \underline{\eta}^2 \Big). \nonumber
		% \end{align}
        Considering $\sqrt{N} - \sqrt{N-1}$ is monotonically decreasing with respect to $N \in \mathbb{N}_+$, it has $\sqrt{N} - \sqrt{N-1} \le \sqrt{1} - \sqrt{1-1} = 1$ indicating $\sqrt{N} \le \sqrt{N-1} + 1$.
        This also means when $i \in \mathbb{S}_2 \cup \mathbb{S}_3$, it is easy to see $c |\tilde{x}_i| > \sqrt{N} \underline{\eta}$.
        Using the results in \cref{property_trigger}, the value of $\| \tilde{\bm{x}} \|^2$ is bounded as
        \begin{align} \label{eqn_c_tilde_x_bound}
			c^2 \| \tilde{\bm{x}} \|^2 &= c^2 \sum_{j = 1}^3 \sum_{i \in \mathbb{S}_j} |\tilde{x}_i|^2 \ge \sum_{i \in \mathbb{S}_2} c^2 |\tilde{x}_i|^2 + \sum_{i \in \mathbb{S}_3} c^2 |\tilde{x}_i|^2 \nonumber \\
			&> \sum_{i \in \mathbb{S}_2} N \underline{\eta}^2 + \sum_{i \in \mathbb{S}_3} \Big( \hat{\eta}_i^2(x_i) + (N - 1) \underline{\eta}^2 \Big) \\
			&= | \mathbb{S}_2 \cup \mathbb{S}_3 | N \underline{\eta}^2 + \sum_{i \in \mathbb{S}_3} \Big( \eta_i^2(x_i) - \underline{\eta}^2 \Big). \nonumber
		\end{align}
		Similarly, the value of $\| \hat{\bm{\eta}}(\bm{x}) \|^2$ is written as
        \begin{align} \label{eqn_eta_bound}
			\| \hat{\bm{\eta}}(\bm{x}) \|^2 &= \sum_{j = 1}^3 \sum_{i \in \mathbb{S}_j} \hat{\eta}_i^2(x_i) \le \sum_{i \in \mathbb{S}_1 \cup \mathbb{S}_2 } \underline{\eta}^2 + \sum_{i \in \mathbb{S}_3} \hat{\eta}_i^2(x_i) \nonumber \\
			&= N \underline{\eta}^2 + \sum_{i \in \mathbb{S}_3} \Big( \hat{\eta}_i^2(x_i) - \underline{\eta}^2 \Big).
		\end{align}
		Apply the inequalities \eqref{eqn_c_tilde_x_bound} and \eqref{eqn_eta_bound} into \eqref{eqn_dotV}, and then due to the positive $c$, \eqref{eqn_dotV} is bounded by
		\begin{align} \label{eqn_dotV_2}
			\dot{V}_c &\le -\frac{2 \| \tilde{\bm{x}} \|}{c \| \tilde{\bm{x}} \| + \| \hat{\bm{\eta}}(\bm{x}) \| } ( c^2 \| \tilde{\bm{x}} \|^2 - \| \hat{\bm{\eta}}(\bm{x}) \|^2) \nonumber \\
			&< -\frac{2 \| \tilde{\bm{x}} \|}{c \| \tilde{\bm{x}} \| + \| \hat{\bm{\eta}}(\bm{x}) \| } ( | \mathbb{S}_2 \cup \mathbb{S}_3 | N \underline{\eta}^2 - N \underline{\eta}^2 ).
		\end{align}
		Now, we consider the case with $\| \tilde{\bm{x}} \| > 2 c^{-1} N \underline{\eta}$.
		This indicates at least one agent satisfying $|\tilde{x}_i|  > 2 c^{-1} \sqrt{N} \underline{\eta}$, which is proven by considering the contradiction with all agents having $c |\tilde{x}_i| \le 2 \sqrt{N} \underline{\eta}$ and
		\begin{align}
			\| \tilde{\bm{x}} \|^2 \le \sum_{i \in \mathcal{V}} ~ 4 c^{-2} N \underline{\eta}^2 = (2 c^{-1} \sqrt{N} \underline{\eta})^2.
		\end{align}
        Note that $2 \sqrt{N} \ge \sqrt{N - 1} + 1$ for $N \in \mathbb{N}$, therefore at least one agent belongs to $\mathbb{S}_2 \cup \mathbb{S}_3$.
        This also means when $\| \tilde{\bm{x}} \| > c^{-1} N \underline{\eta}$ the inequality $| \mathbb{S}_2 \cup \mathbb{S}_3 | \ge 1$ holds, which is then applied in \eqref{eqn_dotV_2} such that $\dot{V}_c$ is bounded as
		\begin{align}
			\dot{V}_c < -\frac{2 \| \tilde{\bm{x}} \| N \underline{\eta}^2}{c \| \tilde{\bm{x}} \| + \| \hat{\bm{\eta}}(\bm{x}) \| } ( | \mathbb{S}_2 \cup \mathbb{S}_3 |  - 1 ) \le 0.
		\end{align}
		This results in the decay of Lyapunov function $V$ when $\| \tilde{\bm{x}} \| > 2 c^{-1} N \underline{\eta}$, indicating $2 c^{-1} N \underline{\eta}$ is the ultimate bound for $\| \tilde{\bm{x}} \|$, which concludes the proof.

		% This also means $| \mathbb{S}_2 \cup \mathbb{S}_3 | \ge 1$ when $\| \tilde{\bm{x}} \| > c^{-1} N \underline{\eta}$, and then the difference between $c^2 \| \tilde{\bm{x}} \|^2$ and $\| \hat{\bm{\eta}}(\bm{x}) \|^2$ is written as
		% \begin{align}
			% 	c^2 \| \tilde{\bm{x}} \|^2 - \| \hat{\bm{\eta}}(\bm{x}) \|^2 > (| \mathbb{S}_2 \cup \mathbb{S}_3 | - 1) N \underline{\eta}^2 \ge 0,
			% \end{align}
		% resulting in $c^2 \| \tilde{\bm{x}} \|^2 > \| \hat{\bm{\eta}}(\bm{x}) \|^2$ and equivalent to $c \| \tilde{\bm{x}} \| > \| \hat{\bm{\eta}}(\bm{x}) \|$.
		% With this result, the negativity of $\dot{V}_c$ is shown for $\| \tilde{\bm{x}} \| > c^{-1} N \underline{\eta}$, indicating $c^{-1} N \underline{\eta}$ is the ultimate bound for $\| \tilde{\bm{x}} \|$, which concludes the proof.
	\end{proof}
	
	\begin{remark} \label{remark_instability_conventional_control_law}
		The proof for \cref{lemma_tracking} also indicates the conventional control law \eqref{eqn_conventional_control_law} cannot achieve average consensus under non-vanishing prediction error $\bm{f}(\bm{x}) - \hat{\bm{\mu}}(\bm{x})$, since the minimal eigenvalue of $\bm{\mathcal{L}}$ is zero inducing infinite large $\epsilon$ by considering $\epsilon = (c \underline{\lambda} (\bm{\mathcal{L}}))^{-1} N \underline{\eta} \to \infty$ from \eqref{eqn_dotV} with $\bm{\mathcal{L}} + \bm{I}_N$ replaced by $\bm{\mathcal{L}}$.
	\end{remark}

	Apply the intermediate result from \cref{lemma_observer} and \cref{lemma_tracking}, the proof of \cref{theorem_consensus_error} is ready.
	
	\textit{Proof of \cref{theorem_consensus_error}:}
	Using the triangular inequality, the average consensus error $\| \bm{x} \!-\! \bm{1}_N \bar{x}^* \|$ is ultimately bounded by \looseness=-1
	\begin{align}
		\lim_{t \to \infty} \| \bm{x} - \bm{1}_N \bar{x}^* \| &\le \lim_{t \to \infty} \| \tilde{\bm{x}} \| + \lim_{t \to \infty} \| \bar{\bm{x}} - \bm{1}_N \bar{x}^* \| \nonumber \\
		&= \lim_{t \to \infty} \| \tilde{\bm{x}} \| \le 2 c^{-1} N \underline{\eta} = \epsilon,
	\end{align}
	inherited the probability of at least $1 - N \delta$ from \cref{lemma_tracking} by considering the probability of \cref{lemma_observer} as $1$.
	\ProofEndSymbol
	
	\cref{theorem_consensus_error} shows the ultimate boundness of the consensus error with high probability, where the consensus error bound $\epsilon$ is increasing with the number of agents $N$.
	This relationship is intuitive, because in the worst case the prediction errors of the individual Gaussian process model on each body are superimposed, leading to a larger prediction error $\hat{\bm{\eta}}(\bm{x})$ with more agents.
	Moreover, the error $\epsilon$ is also related to the quality of the collected data reflected by $\underline{\eta}$, which is proportional to the variance of measurement noise $\sigma_n$ considering \cref{lemma_GP_bound} and \eqref{eqn_sigma_after_update}.
	While using more accurate sensors or processing methods reduces measurement noise, the variance $\sigma_n$ will not vanish in practice.
	To achieve arbitrary small consensus error $\epsilon$, a high gain $c$ is necessary.
	The size of the ultimate bound $\epsilon$ only depends on the control gain $c$, while the auxiliary gain $\bar{c}$ determines the convergence speed \cite{godsil2001algebraic}. 
	This means larger $\bar{c}$ results in faster convergence, which is obtained by considering the convergence speed of the auxiliary state in \cref{lemma_observer} with \cite{godsil2001algebraic} and the fact that \cref{lemma_tracking} only considers $\tilde{x}_i = x_i - \bar{x}_i$.
    The decoupled effects of $c$ and $\bar{c}$ provide us more flexibility to design the controller for desired performance in the steady state, i.e., for smaller ultimate average consensus errors, and the transition phase, i.e., for faster convergence, separately.
	
	\begin{remark} \label{remark_domain_hold}
		Note that only with $x_i(t) \in \mathbb{X}, \forall t \in \mathbb{R}_{0,+}$, the boundness of the prediction error in \cref{lemma_GP_bound} and thus in \cref{theorem_consensus_error} holds.
		Moreover, it is easily to prove the trajectories to $\bar{x}^*$ for each agent are in $\mathbb{X}$ due to the convex domain of $\mathbb{X}$ and $x_(0) \in \mathbb{X}$ for $\forall i \in \mathcal{V}$.
		Therefore, \cref{theorem_consensus_error} holds with safety guarantee for the controlled system, if the neighbor domain of trajectories to $\bar{x}^*$ is a subset of $\mathbb{X}$.
		This can be easily realized by restraining the initial condition $x_i(0)$ to a smaller subset $\mathbb{X}_0 \subset \mathbb{X}$ such that the domain of $x(\infty)$, i.e., $ \{ x \in \mathbb{R}: |x - \bar{x}^* | \le \epsilon \}$, is a subset of $\mathbb{X}$ with $\epsilon = 2 c^{-1} N \bar{\eta}$ defined in \cref{theorem_consensus_error} and $\bar{\eta} = \max_{i \in \mathcal{V}, x_i \in \mathbb{X}} ~ \eta_i(x_i)$ from \cite{yangDistributedLearningConsensus2021}.
	\end{remark}
	
	% \begin{remark}
	% 	The tracking error bound $\epsilon$ is only related to control gain $c$ but not $\bar{c}$ for auxiliary dynamics.
	% 	The value of $\bar{c}$ determines the convergence speed of the auxiliary states \cite{godsil2001algebraic}, which then affects the convergence speed of the average consensus error.
	% 	The decoupled effects of $c$ and $\bar{c}$ provide us more flexibility to design the controller for desired performance in the steady state, i.e., for smaller ultimate average consensus errors, and the transition phase, i.e., for faster convergence, separately.
	% \end{remark}
	
	\begin{remark}
		With the definition of $\epsilon$ in \cref{theorem_consensus_error}, the trigger condition in \eqref{eqn_trigger_function} can be relaxed as
		\begin{align}
			\eta_i(x_i) > c^{-1} \max \left\{ |\tilde{x}_i| - \epsilon / \sqrt{N}, 0 \right\} + \underline{\eta},
		\end{align}
		which is composed of individual consensus error $\tilde{x}_i$, best prediction performance $\underline{\eta}$ and overall consensus error bound $\epsilon$.
		This formulation indicates a general event-trigger design for online learning in multi-agent system, namely comparison of the current prediction error $\eta_i(x_i)$, the deviation of the control error to its bound $|\tilde{x}_i| - \epsilon / \sqrt{N}$, and the guaranteed best prediction performance $\underline{\eta}$, which is also suitable for cooperative learning with different hyperparameters for each Gaussian process model as in \cite{yangDistributedLearningConsensus2021}.
	\end{remark}
	
	%%%%%%%%%%%%%%%%%%%%%%%%%%%%%%%%%%%%%%%%%%%%%%%%%%%%%%%%%%%%%%%%%%%%%%%%%%%%%%%%
	\section{Numerical Simulations} \label{section_simulation}
	
	\subsection{Simulation Setting} \label{subsection_simulation_setting}
	To evaluate the designed control law \eqref{eqn_controller}, \eqref{eqn_observer_dynamics} with event-triggered online learning with \eqref{eqn_general_trigger_formulation}, \eqref{eqn_trigger_function}, we consider a multi-agent system with $N = 4$ agents under dynamics \eqref{eqn_dynamics} with $h(x_i) = 0$, $g(x_i) = 1$ satisfying \cref{assumption_g} and
	\begin{align}
		f(x_i) = \sin(10 x_i) + \frac{1}{2 \exp(-x_i / 10)} + 5
	\end{align}
	for each agent $i \in \mathcal{V} = \{ 1, \cdots, 4 \}$ with the compact domain $\mathbb{X} = [-1.5, 1.5]$.
	The communication topology among the agents is defined by a connected undirected graph $\mathcal{G}$ with the edge set $\mathcal{E} = \{ (1,2), (2,3), (3,4), (4,1) \}$.
	To infer the unknown function $f(\cdot)$, each agent $i \in \mathcal{V}$ employs a Gaussian process model using the kernel function satisfying \cref{assumption_GP} as
    \begin{align}
        \kappa(x, x') = \sigma_f^2 \exp \left( - \frac{|x - x'|^2}{2 l^2} \right)
    \end{align}
    with $\sigma_f = 1$, $l = 0.05$ and the individual data set $\mathbb{D}_i$ satisfying \cref{assumption_dataset} with the variance of the measurement noise as $\sigma_n = 0.01$.
	The simulation time is set as $10$.
	
	To show the effectiveness of the proposed controller with event-triggered online learning, the following cases with different controllers and learning strategies are compared:\looseness=-1
	\begin{enumerate} %[label=(\alph*)]
		\item 
		Use the conventional distributed control law \eqref{eqn_conventional_control_law} for average consensus similarly as in \cite{godsil2001algebraic} with merely offline learning;
		\item 
		Use the conventional distributed control law \eqref{eqn_conventional_control_law} with event-triggered online learning using naive trigger function $\rho_{\text{naive}}(\cdot)$ for all $i \in \mathcal{V}$ with
		\begin{align} \label{eqn_naive_trigger_function}
			\rho_{\text{naive}}(x_i, \bar{x}_i) = \eta_i(x_i) - \underline{\eta};
		\end{align}
		\item 
		Use the proposed controller \eqref{eqn_controller} and \eqref{eqn_observer_dynamics} with merely offline learning;
		\item 
		Use the proposed controller \eqref{eqn_controller} and \eqref{eqn_observer_dynamics} with event-triggered online learning using \eqref{eqn_general_trigger_formulation} and \eqref{eqn_trigger_function};
	\end{enumerate}
	Note that the control gains $c$ in \eqref{eqn_controller} and \eqref{eqn_conventional_control_law} are set as $c = 1$, and the auxiliary gain $\bar{c} = 1$ for \eqref{eqn_observer_dynamics}.
	Moreover, for online learning, the data sets $\mathbb{D}_i$ are initialized with $0$ samples.
	For the offline learning cases, each agent employs an individual data set $\mathbb{D}_i$ including $150$ training samples evenly distributed in the compact domain $\mathbb{X}$.
	Note that the number $150$ for the size of the offline data set is chosen, such that the number of overall collected data using event-triggered online learning in cases (b) and (d) is below $150$ with high probability, as shown in \cref{figure_DataQuantity}.

	\subsection{Performance among Difference Controllers}
	
	First, we show the performance for each cases with $\bm{x}(0) = [-0.52, 0.15,-0.06, -0.71]^T$ resulting in $\bar{x}^* = -0.285$.

	\begin{figure}[t] 
		\centering
		\vspace{0.3cm}
		\begin{tikzpicture}
			\def\file{fig/State.txt}
			\begin{axis}[xlabel={(a)},ylabel={State},
				xmin=0.01, ymin = -0.79, xmax = 9.99,ymax=0.19,legend columns=4,
				width=0.285\textwidth,height=3.5cm,legend pos= south east,
				xlabel shift=-0.2cm,
				xticklabels={,,,}
				]
				\addplot[tab10_blue, thick]	table[x = t_set_ff , y  = x_set_ff_1 ]{\file};
				\addplot[tab10_orange, thick]	table[x = t_set_ff , y  = x_set_ff_2 ]{\file};
				\addplot[tab10_green, thick]	table[x = t_set_ff , y  = x_set_ff_3 ]{\file};
				\addplot[tab10_red, thick]	table[x = t_set_ff , y  = x_set_ff_4 ]{\file};
				\addplot[black, dashed]	table[x = t_set_equilibrium , y  = x_bar_star ]{\file};
			\end{axis}
			\begin{axis}[xlabel={(b)},ylabel={},
				xmin=0.01, ymin = -0.79, xmax = 9.99,ymax=0.19,legend columns=5,
				width=0.285\textwidth,height=3.5cm,
				legend style={at={(0.2,1.05)},anchor=south},
				xlabel shift=-0.2cm,
				xshift=3.6cm,
				xticklabels={,,,}, yticklabels={,,,}
				]
				\addplot[tab10_blue, thick]	table[x = t_set_ft , y  = x_set_ft_1 ]{\file};
				\addplot[tab10_orange, thick]	table[x = t_set_ft , y  = x_set_ft_2 ]{\file};
				\addplot[tab10_green, thick]	table[x = t_set_ft , y  = x_set_ft_3 ]{\file};
				\addplot[tab10_red, thick]	table[x = t_set_ft , y  = x_set_ft_4 ]{\file};
				\addplot[black, dashed]	table[x = t_set_equilibrium , y  = x_bar_star ]{\file};
				\legend{$x_1$ ,$x_2$, $x_3$, $x_4$, $\bar{x}^*$}
			\end{axis}
			\begin{axis}[xlabel={(c) {\color{white} 1} $\!\!t$},ylabel={State},
				xmin=0.01, ymin = -0.79, xmax = 9.99,ymax=0.19,legend columns=4,
				width=0.285\textwidth,height=3.5cm,legend pos= south east,
				xlabel style={text width=0.4cm},
				xshift=0cm, yshift=-2.6cm,
				]
				\addplot[tab10_blue, thick]	table[x = t_set_tf , y  = x_set_tf_1 ]{\file};
				\addplot[tab10_orange, thick]	table[x = t_set_tf , y  = x_set_tf_2 ]{\file};
				\addplot[tab10_green, thick]	table[x = t_set_tf , y  = x_set_tf_3 ]{\file};
				\addplot[tab10_red, thick]	table[x = t_set_tf , y  = x_set_tf_4 ]{\file};
				\addplot[black, dashed]	table[x = t_set_equilibrium , y  = x_bar_star ]{\file};
			\end{axis} 
			\begin{axis}[xlabel={(d) {\color{white} 1} $\!\!t$},ylabel={},
				xmin=0.01, ymin = -0.79, xmax = 9.99,ymax=0.19,legend columns=3,
				width=0.285\textwidth,height=3.5cm,legend pos= south east,
				xlabel style={text width=0.4cm},
				xshift=3.6cm, yshift=-2.6cm,
				yticklabels={,,,}
				]
				\addplot[tab10_blue, thick]	table[x = t_set_tt , y  = x_set_tt_1 ]{\file};
				\addplot[tab10_orange, thick]	table[x = t_set_tt , y  = x_set_tt_2 ]{\file};
				\addplot[tab10_green, thick]	table[x = t_set_tt , y  = x_set_tt_3 ]{\file};
				\addplot[tab10_red, thick]	table[x = t_set_tt , y  = x_set_tt_4 ]{\file};
				\addplot[black, dashed]	table[x = t_set_equilibrium , y  = x_bar_star ]{\file};
			\end{axis}  
		\end{tikzpicture}
        \vspace{-0.3cm}
		\caption{
			The average state $\bar{x}^*$ and the actual state $x_i$ for each agent over time using different controllers and learning strategies, which are corresponding to the $4$ cases in \cref{subsection_simulation_setting}.
		}
        \vspace{-0.3cm}
		\label{figure_State}
	\end{figure}
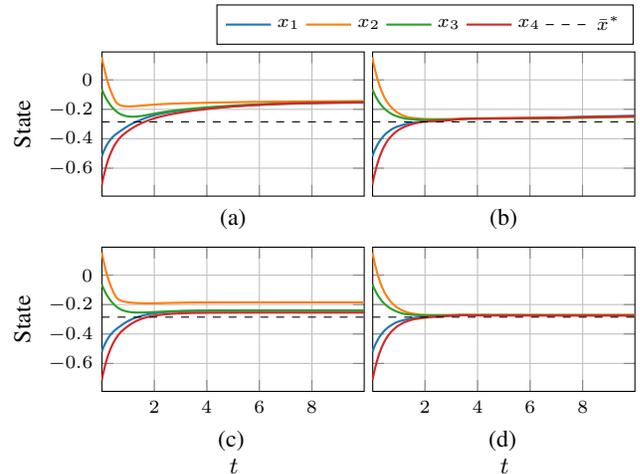
	
	The state evolution for each agent over time is shown in \cref{figure_State}.
	The performance by using the conventional controller \eqref{eqn_conventional_control_law} is shown in subfigures (a) and (b), where the states almost achieve consensus with some errors but do not converge to the average states $\bar{x}^*$.
	The bias between the actual converge state $\bar{x}_{\infty} = \lim_{t \to \infty} \bm{1}_N^T \bm{x}(t)$ and the average states $\bar{x}^*$ can hardly be quantified.
	Therefore, it is not suitable for safety critical scenarios, since the actual converge state $\bar{x}_{\infty}$ may not in the pre-defined domain $\mathbb{X}$.
	Moreover, it is intuitive that with online learning the actual converge state is closer to the average states, because of the smaller prediction error compared to the offline learning.
	With the proposed control using \eqref{eqn_controller} and \eqref{eqn_observer_dynamics}, the bias between $\bar{x}_{\infty}$ and $\bar{x}^*$ is guaranteed to be reduced and shown in subfigures (c) and (d) in \cref{figure_State}, guarantee the effectiveness of \cref{theorem_consensus_error} as in \cref{remark_domain_hold}.
	While the states $x_i$ tend to close $\bar{x}^*$ in (c), the bias to $\bar{x}^*$, i.e., $| \bar{x}_{\infty} - \bar{x}^* |$, is obvious due to the poor prediction under offline learning.
	Additionally, with online learning as in (d), the controlled system achieves a smaller average consensus error compared to case (c).

	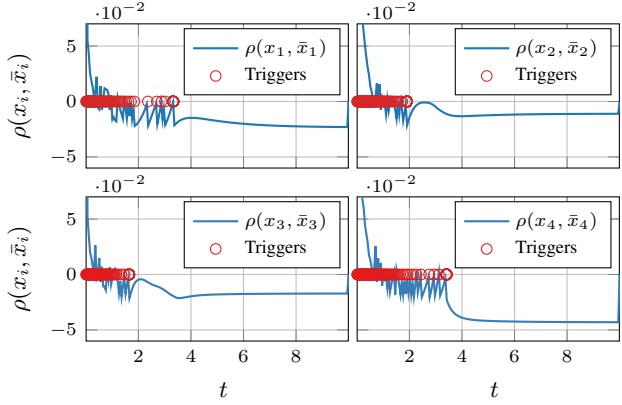
\begin{figure}[t] 
		\centering
        \vspace{0.3cm}
		\begin{tikzpicture}
			\def\file{fig/TriggerRho.txt}
			\begin{axis}[
				xlabel={},ylabel={$\rho(x_i,\bar{x}_i)$},
				xmin=0.01, ymin = -0.06, xmax = 9.99,ymax=0.07,legend columns=1,
				width=0.285\textwidth,height=3.5cm,legend pos= north east,
				xlabel shift=-0.2cm,
				xticklabels={,,,}
				]
				\addplot[tab10_blue, thick]						table[x = t_set_tt , y  = rho_set_plot_1 ]{\file};
				% \addplot[black, dashed]				table[x = t_set_tt , y  = rho_bar ]{\file};
				\addplot[only marks, mark=o, tab10_red]	table[x = t_set_tt_triggerPlot_1 , y  = Trigger_rho_set_triggerPlot_1 ]{\file};
				\legend{{$\rho(x_1, \bar{x}_1)$} , Triggers}
			\end{axis}
			\begin{axis}[
				xlabel={},ylabel={},
				xmin=0.01, ymin = -0.06, xmax = 9.99,ymax=0.07,legend columns=1,
				width=0.285\textwidth,height=3.5cm,legend pos= north east,
				xlabel shift=-0.2cm,
				xshift=3.6cm,
				xticklabels={,,,}, yticklabels={,,,}
				]
				\addplot[tab10_blue, thick]						table[x = t_set_tt , y  = rho_set_plot_2 ]{\file};
				% \addplot[black, dashed]				table[x = t_set_tt , y  = rho_bar ]{\file};
				\addplot[only marks, mark=o, tab10_red]	table[x = t_set_tt_triggerPlot_2 , y  = Trigger_rho_set_triggerPlot_2 ]{\file};
				\legend{{$\rho(x_2, \bar{x}_2)$} , Triggers}
			\end{axis}
			\begin{axis}[
				xlabel={{\color{white} 1111}$t$},ylabel={$\rho(x_i,\bar{x}_i)$},
				xmin=0.01, ymin = -0.06, xmax = 9.99,ymax=0.07,legend columns=1,
				width=0.285\textwidth,height=3.5cm,legend pos= north east,
				xlabel style={text width=1.2cm},
				xshift=0cm, yshift=-2.3cm,
				]
				\addplot[set_blue, thick]						table[x = t_set_tt , y  = rho_set_plot_3 ]{\file};
				% \addplot[black, dashed]				table[x = t_set_tt , y  = rho_bar ]{\file};
				\addplot[only marks, mark=o, set_red]	table[x = t_set_tt_triggerPlot_3 , y  = Trigger_rho_set_triggerPlot_3 ]{\file};
				\legend{{$\rho(x_3, \bar{x}_3)$} , Triggers}
			\end{axis} 
			\begin{axis}[
				xlabel={{\color{white} 1111}$t$},ylabel={},
				xmin=0.01, ymin = -0.06, xmax = 9.99,ymax=0.07,legend columns=1,
				width=0.285\textwidth,height=3.5cm,legend pos= north east,
				xlabel style={text width=1.2cm},
				xshift=3.6cm, yshift=-2.3cm,
				yticklabels={,,,}
				]
				\addplot[set_blue, thick]						table[x = t_set_tt , y  = rho_set_plot_4 ]{\file};
				% \addplot[black, dashed]				table[x = t_set_tt , y  = rho_bar ]{\file};
				\addplot[only marks, mark=o, set_red]	table[x = t_set_tt_triggerPlot_4 , y  = Trigger_rho_set_triggerPlot_4 ]{\file};
				\legend{{$\rho(x_4, \bar{x}_4)$} , Triggers}
			\end{axis}  
		\end{tikzpicture}
        \vspace{-0.3cm}
		\caption{
			The value of trigger function $\rho(x_i, \bar{x}_i)$ and the time of the trigger events for each agent $i \in \mathcal{V}$ for case (d).
		}
        \vspace{-0.3cm}
		\label{figure_Trigger}
	\end{figure}
	
	\cref{figure_Trigger} shows the event times for data collection for online learning with the values of trigger functions $\rho$ in each agent.
	It is obvious that most data collections occur in the transition status as shown in \cref{figure_State} before $t = 4$ , and in the steady states, i.e., $t > 4$ no trigger is activated indicating the Gaussian process models in each agent is sufficiently accurate to achieve the desired $\epsilon$-average consensus.
	In particular, totally $82$, $100$, $71$ and $112$ trigger events occur on agent $1$ to $4$ respectively, resulting in smaller training data sets compared to the offline setting with $|\mathbb{D}_i| = 150$.
	
	\subsection{Monte Carlo Test}
	
	To show the generalization of the proposed controller with event-triggered learning, a Monte Carlo test is implemented, in which each case in \cref{subsection_simulation_setting} is repeated $100$ times with evenly randomly selected initial states $x_i(0) \in \mathbb{X}$ for $\forall i \in \mathcal{V}$ and random initial data sets for case (a) and (c).
	
	\begin{figure}
		\centering
		% \vspace{0.3cm}
		\begin{tikzpicture}
			\def\file{fig/MonteCarloError.txt}
			\begin{semilogyaxis}[xlabel={$t$},ylabel={$\| \bm{x}(t) - \bm{1}_N \bar{x}^* \|$},
				xmin=0.01, ymin = 7e-3, xmax = 9.99,ymax=30,legend columns=4,
				width=0.48\textwidth,height=4.5cm,legend pos=north east]
				\addplot[tab10_blue, thick]    table[x = t_set_ff , y  = norm_e_set_ff_all_mean ]{\file};
				\addplot+[name path=max_ff,black,no markers, draw=none] table[x = t_set_ff , y  = norm_e_set_ff_all_max ]{\file};
				\addplot+[name path=min_ff,black,no markers, draw=none] table[x = t_set_ff , y  = norm_e_set_ff_all_min ]{\file};
				\addplot[tab10_blue!20] fill between[of=max_ff and min_ff];
				
				\addplot[tab10_orange, thick]    table[x = t_set_ft , y  = norm_e_set_ft_all_mean ]{\file};
				\addplot+[name path=max_ft,black,no markers, draw=none] table[x = t_set_ft , y  = norm_e_set_ft_all_max ]{\file};
				\addplot+[name path=min_ft,black,no markers, draw=none] table[x = t_set_ft , y  = norm_e_set_ft_all_min ]{\file};
				\addplot[tab10_orange!20] fill between[of=max_ft and min_ft];
				
				\addplot[tab10_green, thick]    table[x = t_set_tf , y  = norm_e_set_tf_all_mean ]{\file};
				\addplot+[name path=max_tf,black,no markers, draw=none] table[x = t_set_tf , y  = norm_e_set_tf_all_max ]{\file};
				\addplot+[name path=min_tf,black,no markers, draw=none] table[x = t_set_tf , y  = norm_e_set_tf_all_min ]{\file};
				\addplot[tab10_green!20] fill between[of=max_tf and min_tf];
				
				\addplot[tab10_red, thick]    table[x = t_set_tt , y  = norm_e_set_tt_all_mean ]{\file};
				\addplot+[name path=max_tt,black,no markers, draw=none] table[x = t_set_tt , y  = norm_e_set_tt_all_max ]{\file};
				\addplot+[name path=min_tt,black,no markers, draw=none] table[x = t_set_tt , y  = norm_e_set_tt_all_min ]{\file};
				\addplot[tab10_red!20] fill between[of=max_tt and min_tt];
				
				\legend{
					(a),,, {Variance},
					(b),,, {Variance},
					(c),,, {Variance},
					(d),,, {Variance}
				}
			\end{semilogyaxis}
		\end{tikzpicture}
		\vspace{-0.3cm}
		\caption{
			Average consensus error over time with variance.
		}
		\vspace{-0.3cm}
		\label{figure_Error}
	\end{figure}
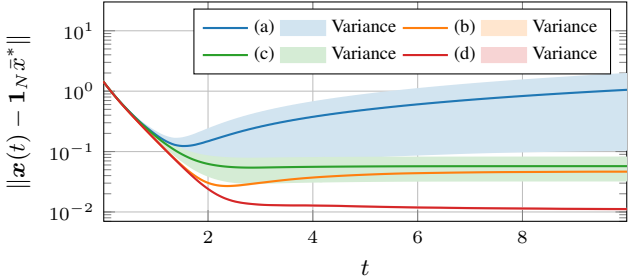 
	
	The average consensus error $\| \bm{x} - \bm{1}_N \bar{x}^* \|$ for each case is shown in \cref{figure_Error}, where the cases (a) and (c) with merely offline learning induce larger variance, showing strong dependency on the initial training sets.
	In comparison, two online learning approaches (b) and (d) demonstrate smaller average consensus error with smaller variance, because the training sets created through online data collection have desired prediction accuracy for the average consensus problem, i.e., $\hat{\eta}_i(\cdot) \le \underline{\eta}$ as in \eqref{eqn_sigma_after_update}.
	Moreover, there exists a divergence trend for two cases (a) and (b) with conventional distributed control law by considering the increasing average consensus error, indicating the observations in \cref{figure_State} (a) and (b) are common.
	Due to the larger bias between actual convergence state $\bar{x}_{\infty}$ and $\bar{x}^*$ by using \eqref{eqn_conventional_control_law}, it is intuitive that cases (a) and (b) induce larger average consensus errors than case (c) and (d) with the proposed controller.

	\begin{figure}[t] 
		\centering
		\vspace{0.3cm}
		\begin{tikzpicture}
			\def\file{fig/MonteCarloDataQuantity.txt}
			\begin{axis}[ylabel style={align=center},ylabel={Maximum Size \\ of the Data Set},
				ymin=1, ymax=260,legend columns=2,
				xmin=0.5,xmax=4.5,
				width=0.48\textwidth,height=4cm,legend pos= north east,
				ylabel shift = -0.1cm,  xshift=0cm,
				ybar,
				xtick={1,2,3,4},xticklabel style={align=center}, 
				xticklabels = {{Agent 1},{Agent 2},{Agent 3},{Agent 4}},
				bar width=0.3cm,ybar=0.0cm,
				legend image code/.code={
					\draw [#1] (0cm,-0.1cm) rectangle (0.4cm,0.1cm); },]
				\addplot[fill=set_red!30]    table[x = AgentNrSet , y  = GP_DataQuantitySet_tf_mean ]{\file};
				\addplot[fill=set_blue!30]    table[x = AgentNrSet , y  = GP_DataQuantitySet_ft_mean ]{\file};
				\addplot[fill=set_green!30]    table[x = AgentNrSet , y  = GP_DataQuantitySet_tt_mean ]{\file};
				
				\addplot [black, only marks, mark=.] plot [error bars/.cd, y dir=both, y explicit relative]
				table [x = AgentNrSet , y  = GP_DataQuantitySet_tf_mean, y error plus=GP_DataQuantitySet_tf_var, y error minus=GP_DataQuantitySet_tf_var] {\file};
				\addplot [set_blue, only marks, mark=.] plot [error bars/.cd, y dir=both, y explicit relative, error bar style={line width=0.8pt, solid}]
				table [x = AgentNrSet , y  = GP_DataQuantitySet_ft_mean, y error plus=GP_DataQuantitySet_ft_var, y error minus=GP_DataQuantitySet_ft_var] {\file};
				\addplot [set_green, only marks, mark=.] plot [error bars/.cd, y dir=both, y explicit relative, error bar style={line width=0.8pt, solid}]
				table [x = AgentNrSet , y  = GP_DataQuantitySet_tt_mean, y error plus=GP_DataQuantitySet_tt_var, y error minus=GP_DataQuantitySet_tt_var] {\file};
				\legend{{Offline learning (a)\&(c)},{Online learning (b)},{Online learning (d)}}
			\end{axis}
		\end{tikzpicture}
        \vspace{-0.3cm}
		\caption{
			Maximal size of data set, which for offline learning in cases (a) and (c) is $150$ from the size of initial data set.
			With naive event-triggered online learning \eqref{eqn_naive_trigger_function} in case (b), the maximal size of the data set for each agent $1$ to $4$ denotes $104 \pm 42$, $98 \pm 43$, $108 \pm 42$ and $104 \pm 43$, respectively.
			The maximal numbers of training samples in case (d) collected through decentralized event-triggered mechanism with \eqref{eqn_trigger_function} are $102 \pm 41$, $97 \pm 42$, $107 \pm 41$ and $102 \pm 42$ for agent $1$ to $4$, respectively.
		}
        \vspace{-0.3cm}
		\label{figure_DataQuantity}
	\end{figure}
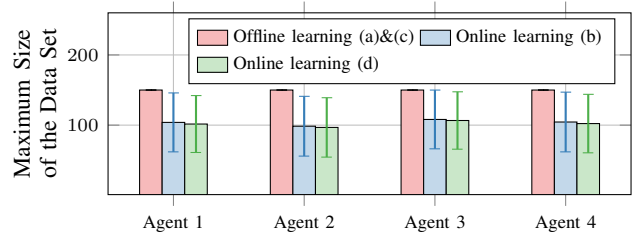
	
	To show the comparison in \cref{figure_Error} makes sense with similar number of training data, the size of the eventual data set in each agent is shown in \cref{figure_DataQuantity}.
	It is shown both event-triggered online learning strategies results in a data set with average near $100$ samples and almost bounded by $150$ samples.
	This indicates the predefined size of initial data set for cases (a) and (b) is sufficient large for data efficiency comparison with event-triggered online learning in cases (b) and (d).
	Combining with the results in \cref{figure_Error}, it is demonstrated that with similar number of data sets event-triggered online learning performs better in the control performance with smaller average consensus errors.
	
	Until here, all results derived for the proposed distributed controller with decentralized event-triggered online learning in \cref{section_distributed_event_trigger} are observed, indicating the effectiveness of the proposed methods.
	%%%%%%%%%%%%%%%%%%%%%%%%%%%%%%%%%%%%%%%%%%%%%%%%%%%%%%%%%%%%%%%%%%%%%%%%%%%%%%%%
	\section{Conclusion} \label{section_conclusion}
	In this paper, we propose an innovative approach to achieving average consensus in MASs considering unknown dynamics. 
	This approach leverages a distributed consensus control law, which integrates an auxiliary dynamics. 
	Additionally, to enhance control performance in real-time online learning scenarios and ensure high data efficiency, a decentralized event-trigger mechanism is designed for each agent.
	Remarkably, this mechanism operates solely on local state information, obviating the necessity for knowledge of the global communication topology. 
	The theoretical proof of this work substantiates the stability of the system while guaranteeing an upper bound on the average consensus error making it particularly suitable for safety-critical applications. 
	We conclude by demonstrating the effectiveness of our proposed control and learning strategies through comparative numerical simulation.

	% \addtolength{\textheight}{-12cm}   % This command serves to balance the column lengths
	% on the last page of the document manually. It shortens
	% the textheight of the last page by a suitable amount.
	% This command does not take effect until the next page
	% so it should come on the page before the last. Make
	% sure that you do not shorten the textheight too much.

	%%%%%%%%%%%%%%%%%%%%%%%%%%%%%%%%%%%%%%%%%%%%%%%%%%%%%%%%%%%%%%%%%%%%%%%%%%%%%%%%
	\section*{Appendix}
	\label{appendix_an_example}
	
	In the appendix, an example is given to show the states of the system \eqref{eqn_dynamics} controlled by \eqref{eqn_conventional_control_law} are unbounded under bounded disturbance.
	We observe a simple MAS with only $2$ agents under fully connected communication topology.
	Moreover, let the dynamics of each agent follows \eqref{eqn_dynamics}, which is then controlled by \eqref{eqn_conventional_control_law}.
	Then, the concatenated controlled system is written as
	\begin{align} \label{eqn_example_2DoF_controlled_dynamics}
		\begin{bmatrix}
			\dot{x}_1 \\ \dot{x}_2
		\end{bmatrix} = - c \begin{bmatrix}
			1 & -1 \\ -1 & 1
		\end{bmatrix} \begin{bmatrix}
			x_1 \\ x_2
		\end{bmatrix} + \begin{bmatrix}
			f(x_1) - \mu_1(x_1) \\ f(x_2) - \mu_1(x_2)
		\end{bmatrix}
	\end{align}
	with $\mu_i(x_i)$ is sufficiently accurate such that only a small constant prediction error $\varepsilon \in \mathbb{R}_+$ is left, i.e., $f(x_i) - \mu_i(x_i) = \varepsilon$ for all $i = 1, 2$.
	The explicit solution of \eqref{eqn_example_2DoF_controlled_dynamics} denotes
	% \begin{align}
	% 	x_1(t) =  \big(x_1(0) + x_2(0)&\big) / 2 + \varepsilon t \\
	% 	&- \big(x_1(0) - x_2(0)\big) \exp (-2 c t) / 2, \nonumber \\
	% 	x_2(t) =  \big(x_1(0) + x_2(0)&\big) / 2 + \varepsilon t\\
	% 	&+ \big(x_1(0) - x_2(0)\big) \exp (-2 c t) / 2. \nonumber
	% \end{align}
    \begin{align}
		\begin{bmatrix}
			x_1(t) \\ x_2(t)
		\end{bmatrix}\!\! \!=\! \!\!\begin{bmatrix}
			1 \\ 1
		\end{bmatrix}\!\! \frac{x_1\!(0)\! \!+\!\! x_2\!(0)\! \!+\! 2\varepsilon t}{2}\!  - \!\begin{bmatrix}
			1 \\ \!-\!1
		\end{bmatrix}\! \frac{x_1\!(0)\! \!-\! x_2\!(0)\!}{2} \exp (\!-\!2 c t)
	\end{align}
	For non-zero prediction error $\varepsilon$, the unboundness of the state trajectory is obvious to see by considering
	\begin{align}
		\lim_{t \to \infty} x_i(t) = \big(x_1(0) + x_2(0)\big) / 2 + \varepsilon \lim_{t \to \infty} t \to \infty,
	\end{align}
	for $i = 1, 2$.
	Therefore, the conventional control law \eqref{eqn_conventional_control_law} can only stabilize the MAS with \eqref{eqn_dynamics} for completely known dynamics, i.e., $f(x_i) = \mu_i(x_i)$, and not suitable in our setting.
	There also exists other robust or adaptive control laws for average consensus, which will be regarded as future work due to our main focus on the decentralized event-trigger design for online learning in multi-agent systems.
	
		%%%%%%%%%%%%%%%%%%%%%%%%%%%%%%%%%%%%%%%%%%%%%%%%%%%%%%%%%%%%%%%%%%%%%%%%%%%%%%%%
	
% \section*{Acknowledgment}
% 		This work has been financially supported by the Federal Ministry of Education and Research of Germany in the programme of ``Souverän. Digital. Vernetzt.'' under joint project 6G-life with project identification number: 16KISK002 and the Germany Federal Ministry of Health (BMG) under grant No. 2523DAT400 (project ``AI-assisted analysis and visualization of pandemic situations'' | AI-DAVis-PANDEMICS).
	
	%%%%%%%%%%%%%%%%%%%%%%%%%%%%%%%%%%%%%%%%%%%%%%%%%%%%%%%%%%%%%%%%%%%%%%%%%%%%%%%%
	
	%References are important to the reader; therefore, each citation must be complete and correct. If at all possible, references should be commonly available publications.
	\bibliography{refs}

\end{document}